\documentclass[11pt]{article}

\usepackage{bbm}
\usepackage{graphicx}
\usepackage{amsmath}
\usepackage{amsthm}
\usepackage{thmtools, thm-restate}
\usepackage{amsfonts}
\usepackage{listings}
\usepackage{color}
\usepackage{microtype}
\usepackage{algorithm}
\usepackage{tcolorbox}
\usepackage{enumitem}
\usepackage{mathrsfs}%
\usepackage[title]{appendix}%
\usepackage{xcolor}%
\usepackage{manyfoot}
\usepackage{booktabs}%
\usepackage[T1]{fontenc}
\usepackage{fullpage}
\usepackage{hyperref}

\newcommand{\R}{\mathbb{R}}
\newcommand{\eps}{\varepsilon}

\newcommand{\ilog}{\textsc{iLog}}
\DeclareMathOperator{\util}{util}

\newtheorem{theorem}{Theorem}[section]

\newtheorem{lemma}[theorem]{Lemma}

\newtheorem{observation}{Observation}[section]
\newtheorem{invariant}{Invariant}

\DeclareMathOperator*{\poly}{poly}

\newtcolorbox[auto counter, number within=section]{algbox}[2][]{%
    title=Algorithm~\thetcbcounter: {#2}, #1}
\newtcolorbox[auto counter, number within=section]{funcbox}[2][]{%
    title= {#2}, #1}

\usepackage{mdframed}
\usepackage{wrapfig}
\newcommand{\erclogowrapped}[1]{%
\setlength\intextsep{0pt}%
\begin{wrapfigure}[3]{r}{#1}%
\includegraphics[width=#1]{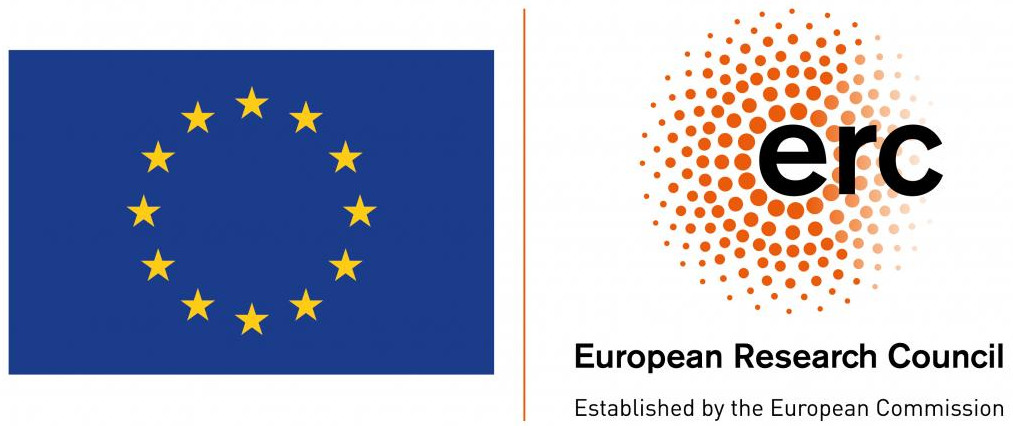}%
\end{wrapfigure}%
}

\newcommand{\myparagraph}[1]{\noindent {\textbf{#1}} }

\title{Multiplicative Auction Algorithm for Approximate Maximum Weight Bipartite Matching}

\author{
Da Wei Zheng 
\\ \small{University of Illinois Urbana-Champaign}
\\ \small{Urbana, IL, USA}
\and
Monika Henzinger
\\ \small{Institute of Science and Technology Austria }
\\ \small{Klosterneuburg, Austria}
}
\date{}

\begin{document}

\maketitle

\begin{abstract}
We present an \emph{auction algorithm}  using multiplicative instead of constant weight updates to compute a $(1-\eps)$-approximate maximum weight matching (MWM) in a bipartite graph with $n$ vertices and $m$ edges in time $O(m\eps^{-1})$, beating the running time of the fastest known approximation algorithm of Duan and Pettie [JACM '14] that runs in $O(m\eps^{-1}\log \eps^{-1})$. 
Our algorithm is very simple and it can be extended to give a dynamic data structure that maintains a $(1-\eps)$-approximate maximum weight matching under (1) one-sided vertex deletions (with incident edges) and (2) one-sided vertex insertions (with incident edges sorted by weight) to the other side.
The total time used is $O(m\eps^{-1})$, where $m$ is the sum of the number of initially existing and inserted edges.\footnote{An earlier version of this paper appeared in the 24th International Conference on Integer Programming and Combinatorial Optimization (IPCO 2023) with a slightly slower algorithm running in $O(m\eps^{-1}\log \eps^{-1})$ time.}
\end{abstract}

\section{Introduction}

Let $G = (U \cup V, E)$ be an edge-weighted bipartite graph with $n = |U\cup V|$ vertices and $m = |E|$ edges where each edge $uv \in E$ with $u\in U$ and $v\in V$ has a non-negative weight $w(uv)$.
The \emph{maximum weight matching} (MWM) problem asks for a matching $M \subseteq E$ that attains the largest possible weight $w(M) = \sum_{uv\in M} w(uv)$.
This paper will focus on approximate solutions to the MWM problem. More specifically, if we let $M^*$ denote a maximum weight matching of $G$, our goal is to find a matching $M$ such that $w(M) \ge (1-\eps) w(M^*)$ for any small constant $\eps >0$.

Matchings are a very well studied problem in combinatorial optimization.
Kuhn \cite{Kuhn1955} in 1955 published a paper that started algorithmic work in matchings, and presented what he called the ``Hungarian algorithm'' which he attributed the work to K\H{o}nig and Egerv\'{a}ry. 
Munkres \cite{Munkres1957} showed that this algorithm runs in $O(n^4)$ time.
The running time for computing the exact MWM has been improved many times since then.
Recently, Chen et al. \cite{ChenKLPPS2022} showed that it was possible to solve the more general problem of max flow in $O(m^{1+o(1)})$ time.

For $(1-\eps)$-approximation algorithms for MWM in bipartite graphs, 
Gabow and Tarjan in 1989 showed an $O(m\sqrt{n}\log(n/\eps))$ algorithm.
Since then there were a number of results for different running times and different approximation ratios.
The prior best approximate algorithm is by Duan and Pettie \cite{DuanP14} which computes a $(1-\eps)$-approximate maximum weight matching in $O(m\eps^{-1}\log(\eps^{-1}))$ time with a scaling algorithm.
We defer to their work for a more thorough survey of the history on the MWM problem.

We show in our work that the auction algorithm for matchings using multiplicative weights can give a $(1-\eps)$-approximate maximum weight matching with a running time of $O(m\eps^{-1})$ for bipartite graphs. This is a modest improvement of a $\log \eps^{-1}$ factor over the prior algorithm of Duan and Pettie \cite{DuanP14} which works in general graphs. 
However, in comparison to their rather involved algorithm, our algorithm is simple and only uses elementary data structures.
Furthermore, we are able to use properties of the algorithm to support two dynamic operations, namely one where vertices are deleted from one side and one where vertices of the other side of the bipartite graph are inserted together with their incident edges.
No algorithm that allows both these operations with running time faster than recomputation from scratch was known prior.

\subsection{Dynamic matching algorithms.}
\myparagraph{Dynamic weighted matching.} There has been a large body of work on dynamic matching and many variants of the problem have been studied, e.g, the maximum, maximal, as well as $\alpha$-approximate setting for a variety of values of $\alpha$, both in the weighted as well as in the unweighted setting. See~\cite{DBLP:conf/sand/HanauerH022} for a survey of the current state of the art for the fully dynamic setting.
For any constant $\delta >0$ there is a conditional lower bound based on the OMv conjecture that shows that
any dynamic algorithm that returns the \emph{exact} value of a maximum cardinality matching in a bipartite graph with polynomial preprocessing time
cannot take time $O(m^{1-\delta})$ per query and $O(m^{1/2 - \delta})$ per edge update operation~\cite{henzinger2015unifying}.

\myparagraph{Dynamic $(1-\eps)$-approximate matchings}
For \emph{general weighted} graphs Gupta and Peng \cite{GuptaP13} gave the first algorithm in the \emph{fully dynamic} setting with edge insertions and deletions to maintain a $(1-\eps)$-approximate matching in $O(\eps^{-1}\sqrt{m} \log w_{max})$ time, where the edges fall into the range $[1,w_{max}]$.
There are also some results for bipartite graphs in \emph{partially dynamic} settings. In the \emph{incremental} setting, edges are only inserted, and \emph{decremental} setting, edges are only deleted.
For unweighted bipartite graphs, the fastest known decremental algorithm is by Bernstein, Probst Gutenberg, and Saranurak~\cite{BernsteinGS20} achieves update times of $O(\eps^{-4}\log^3(n))$ per edge deletion. For incremental algorithms Blikstad and Kiss~\cite{BlikstadK23} achieve update times of $O(\eps^{-6})$ time per edge insertion.
These results can be made to work in weighted graphs by a meta theorem of Bernstein, Dudeja, and Langley~\cite{BernsteinDL21}. Their theorem states that any dynamic algorithm on an unweighted bipartite graph can be transformed into a dynamic algorithm on weighted bipartite graph at the expense of an extra  $(1/\eps)^{O(1/\eps)}\log n$ factor.

\myparagraph{Vertex updates.}
By vertex update we refer to updates that are vertex insertion (resp.~deletion) that also inserts (resp.~deletes) all edges incident to the vertex.
There is no prior work on maintaining matchings in weighted graphs under vertex updates. 
However, vertex updates in the \emph{unweighted bipartite} setting has been studied.
Bosek et al. \cite{BosekLSZ14} gave an algorithm that maintains the $(1-\eps)$-approximate matching when vertices of one side are  deleted  in $O(\eps^{-1})$ amortized time per changed edge. The algorithm can be adjusted to the setting where vertices of one side are inserted in the same running time, but it cannot handle both vertex insertions and deletions.
Le et al. \cite{LeMSW22} gave an algorithm for maintaining a \emph{maximal} matching under vertex updates in constant amortized time per changed edge. They also presented an $e/(e-1) \approx 1.58$ approximate algorithm for maximum matchings in an unweighted graph when vertex updates are only allowed on one side of a bipartite graph.

We give the first algorithm to maintain a $(1-\eps)$-approximate maximum weight matching where vertices can undergo vertex deletions  on one side \emph{and} vertex insertions on the other side 
in total time $O(m\eps^{-1})$, where $m$ is the sum of the number of initially existing and inserted edges. It assumes that the edges incident to an inserted vertex are given in sorted order by weight, otherwise, the running time increases by $O(\log n)$ per inserted edge.

\subsection{Linear Program for MWM}
The MWM problem can be expressed as the following \emph{linear program} (LP) where the variable $x_{uv}$ denotes whether the edge $uv$ is in the matching. It is well known \cite{Schrijver03} that the below LP is integral, that is the optimal solution has all variables $x_{uv}\in \{0, 1\}$.
\begin{align*}
    \max \quad & \sum_{uv\in E} w(uv) x_{uv} &\\
    s.t. \quad &  \sum_{v \in N(u)} x_{uv} \le 1 &\forall u\in U \\
         &  \sum_{u \in N(v)} x_{uv} \le 1 &\forall v\in V \\
         &  x_{uv} \ge 0                      &\forall uv \in E 
\end{align*}

We can also consider the dual problem of weighted vertex cover that aims to find dual weights $y_u$ and $y_v$ for every vertex $u\in U$ and $v\in V$ respectively. 
\begin{align*}
    \min \quad & \sum_{u\in U} y_u + \sum_{v\in V} y_v &\\
    s.t.\quad &  y_u + y_v \ge w(uv)            &\forall uv \in E \\
         &  y_u \ge 0                       &\forall u \in U \\
         &  y_v \ge 0                       &\forall v \in V 
\end{align*}

\subsection{Multiplicative weight updates for packing LPs}
Packing LPs are LPs of the form 
$\max\{c^{T}x \mid Ax \le b\}$
for $c\in \R^{n}_{\ge 0}$, $b\in \R^{m}_{\ge 0}$ and 
$A\in \R^{n\times m}_{\ge 0}$.
The LP for MWM is a classical example of a packing LP.
The \emph{multiplicative weight update method} (MWU) has been investigated extensively to provide faster algorithms for finding approximate solutions\footnote{
By \emph{approximate solution} we mean a possibly fractional assignments of variables that obtains an approximately good LP objective. 
If we find such an approximate solution to MWM, fractional solutions need to be rounded to obtain an actual matching. } 
to packing LPs \cite{Young14, KoufogiannakisY14, ChekuriQ18, Allen-ZhuO19, WangRM16, Quanrud20}. 
Typically the running times for solving these LPs have a dependence on $\eps$ of $\eps^{-2}$,
e.g. the algorithm of Koufogiannakis and Young \cite{KoufogiannakisY14} would obtain a running time of $O(m\eps^{-2}\log n)$ when applied to the matching LP. 

The fastest multiplicative weight update algorithm for solving packing LPs by Allen-Zhu and Orecchia \cite{Allen-ZhuO19} would obtain an $O(m\eps^{-1}\log n)$ running time for MWM. 
Very recently, work by Battacharya, Kiss, and Saranurak \cite{BhattacharyaKS22} extended the MWU for packing LPs to the \emph{partially dynamic setting}. When restricted to the MWM problem means the weight of edges either only increase or only decrease. 
Using similar ideas with MWUs, Assadi~\cite{Assadi23} recently derived a simple semi-streaming algorithm for bipartite matchings.
However as packing LPs are more general than MWM, these algorithms are significantly more complicated and are slower by $\log n$ factors (and sometimes worse dependence on $\eps$ e.g. in \cite{BhattacharyaKS22}) when compared to our algorithms.

We remark that our algorithm, while it uses multiplicative weight updates, is unlike typical MWU algorithms as it has an additional monotonicity property.
We only increase dual variables on one side of the matching. 

\subsection{Auction Algorithms}
Auction algorithms are a class of primal dual algorithms for solving the MWM problem that view $U$ as a set of \emph{goods} to be sold, $V$ as a set of \emph{buyers}. The goal of the auction algorithm is to find a welfare-maximizing allocation of goods to buyers. 
The algorithm is attributed to Bertsekas \cite{Bertsekas81}, as well as to Demange, Gale, and Sotomayor \cite{DemangeGG86}.

An auction algorithm initializes the prices of all the goods $u\in U$ with a price $y_u = 0$ (our choice of $y_u$ is intentional, as prices correspond directly to dual variables), and has buyers initially \emph{unallocated}. 
For each buyer $v\in V$, the \emph{utility} of that buyer upon being allocated $u\in U$ is $\util(uv) = w(uv) - y_u$.
The auction algorithm proceeds by asking an unallocated buyer $v \in V$ for the good they desire that maximizes their utility, i.e. for $u_v = \arg\max_{u\in N(v)} \util(uv)$.
If $\util(u_v v) < 0$, the buyer remains unallocated.
Otherwise the algorithm allocates $u_v$ to $v$, then increases the price $y_u$ to $y_u + \eps$.
The algorithm terminates when all buyers are either allocated or for every unallocated buyer $v$, it holds that $\util(u_v v) < 0$. 
If the maximum weight among all the edges is $w_{max}$, then the auction algorithm terminates after $O(n\eps^{-1}w_{max})$ rounds and outputs a matching that differs from the optimal by an additive factor of at most $n\eps$.

There have been a recent resurgence in interest in auction algorithms. Assadi, Liu, and Tarjan~\cite{AssadiLT21} used the auction algorithm for matchings in unweighted graphs in semi-streaming and massively parallel computing (MPC) settings. This work was generalized for weighted bipartite graphs in the same settings by Liu, Ke, and Kholler~\cite{LiuKK23}.

\subsection{Our contribution}
We present the following modification of the auction algorithm:
\begin{center}
\textit{
When $v$ is allocated $u$, increase $y_u$ to $y_u + \eps \cdot w(uv)$ instead of $y_u + \eps$.}
\end{center}

Note that this decreases $\util(v)$ by at least a factor of $(1-\eps)$ as well as increases $y_u$ by at least a factor of $(1+\eps)$.
Thus we will call algorithms with this modification \emph{multiplicative auction algorithms}.
Surprisingly, we were not able to find any literature on this simple modification.
Changing the constant additive weight update to a multiplicative weight update
has the effect of taking much larger steps when the weights are large, and so we are able to show that the algorithm can have no dependence on the size of the weights. 
In fact, we are able to improve the running time to $O(m\eps^{-1})$, faster than the prior approximate matching algorithm of Duan and Pettie \cite{DuanP14} that ran in $O(m\eps^{-1}\log \eps^{-1})$.
While the algorithm of~\cite{DuanP14} has the advantage that it works for general graphs and ours is limited to bipartite graphs,
our algorithm is simpler as it avoids the scaling algorithm framework and is easier to implement.
\begin{restatable}{theorem}{staticmatching} \label{thm:StaticMatching}
Let $G=(U\cup V, E)$ be a weighted biparitite graph and $\eps$ be a value such that $1>\eps>0$. There is a multiplicative auction algorithm running in time $O(m\eps^{-1})$ that finds a $(1-\eps)$-approximate maximum weight matching of $G$.
\end{restatable}

Furthermore, it is straightforward to extend our algorithm to a setting where \emph{vertices on one side are deleted} and \emph{vertices on the other side are added with all incident edges given in sorted order of weight}. When the inserted edges are not sorted by weight, the running time per inserted edge increases by an additive term of $O(\log \log n)$ to sort the log of the weights of incident inserted edges.
\begin{restatable}{theorem}{decmatching}\label{thm:DecMatching}
Let $G=(U\cup V, E)$ be a weighted bipartite graph.
There exists a dynamic data structure that  maintains a $(1-\eps)$-approximate maximum weight matching of $G$  and supports any arbitrary sequence of the following operations
\begin{enumerate}
\item[(1)] Deleting a vertex in $U$ 
\item[(2)] Adding a new vertex into $V$ with all its incident edges sorted by weight
\end{enumerate}
in total time $O(m\eps^{-1})$, where $m$ is sum of the number of initially existing and inserted edges.
\end{restatable}
\section{The static algorithm}

\subsection{A slower algorithm}\label{subsec:slower}
For sake of exposition we will first present a slower algorithm that runs in near-linear time in the number of edges that will use the following update rule:
\begin{center}
\textit{
When $u$ is allocated to $v$, increase $y_u$ to $y_u + \eps \cdot \util(uv)$ }
\end{center}

We assume that the algorithm is given as input some fixed $0<\eps'<1$, and the goal is to find a $(1-\eps')$-approximate MWM.
We will also assume that $m= \Omega(n)$, as a graph with $m$ edges has at most $2m$ vertices that have at least one incident edge. If $2m < n$, then we may discard the isolated vertices and reduce $n$.

\myparagraph{Notation}
For sake of notation let $N(u) = \{v\in V \mid uv \in E\}$ be the set of neighbors of $u\in U$ in $G$, and similarly for $N(v)$ for $v\in V$. 

\myparagraph{Preprocessing of the weights.}
Let $w_{max} > 0$ be the maximum weight edge of $E$. For our static auction algorithm we may ignore any edge $uv \in E$ of weight less than $\eps'\cdot  w_{max}/n$ as $w(M^*) \ge w_{max}$ as taking $n$ of these small weight edges would not even contribute $\eps'\cdot w(M^*)$ to the matching. 
Thus, we only consider edges of weight at least $\eps'\cdot w_{max}/n$, which allows us to
rescale all edge weights by dividing them by $\eps'\cdot w_{max}/n$. As a result we can assume (by slight abuse of notation) in the following that the minimum edge weight is $1$ and the largest edge weight $w_{max}$ equals $ n/\eps'$.
Furthermore, since we only care about approximations, we will also round down all edge weights to the nearest power of $(1+\eps)$ for some $\eps < \eps'/2$.
We define $\ilog(x) = \lfloor \log_{1+\eps} (x) \rfloor$, and we will only care about weights after applying this operation.

Let $k_{max} = \ilog(w_{max}) = \ilog(n/\eps') = O(\eps^{-1} \log(n/\eps))$.
Let $k_{min}$ be the smallest integer such that
$(1+\eps)^{-k_{min}} \le \eps$.
Observe that as $\log (1 +\eps) \le \eps$ for $0 \le \eps \le 1$ it holds that
\[k_{min} \ge \frac{\log(\eps^{-1})}{\log (1+\eps)} \ge \eps^{-1} \log (\eps^{-1}). \]

Thus we see that $k_{min} = \Theta(\eps^{-1}\log(\eps^{-1}))$.

\myparagraph{Algorithm.}
The algorithm first builds for every $v \in V$ a list $Q_v$ of pairs $(i,uv)$ for each edge $uv$ and each value $i$ with $-k_{min} \le i \le j_{uv} = \ilog(w_{uv})$ and then sorts $Q_v$ by decreasing value of $i$. After, it calls the function 
$\textsc{MatchR}(v)$ on every $v \in V$.
The function
$\textsc{MatchR}(v)$ matches $v$ to the item that maximizes its utility and updates the price $y_u$ according to our multiplicative update rule. While matching $v$, another vertex $v'$ originally matched to $v$ may become unmatched. If this happens, $\textsc{MatchR}(v')$ is called immediately after $\textsc{MatchR}(v)$. 

\begin{algbox}{\textsc{MultiplicativeAuction}$(G = (U\cup V, E))$} \label{alg:MultiplicativeAuction}
$M \leftarrow \emptyset$. \\
$y_u \leftarrow 0$ for all $u \in U$.\\
$y_v \leftarrow 0$ for all $v \in V$ \hfill {\color{gray} \# This is only used in the analysis}
\\
$j_v \leftarrow k_{max}$ for all $v \in V$ \hfill {\color{gray} \# This is only used in the analysis} \\
$Q_v \leftarrow \emptyset$ for all $v \in V$.  \\
For $v\in V$:
\begin{enumerate}
    \item For $u \in N(v)$:
    \begin{enumerate}
        \item $j_{uv} \leftarrow \ilog(w(uv))$
        \item For $i$ from $j_{uv}$ to $-k_{min}$:
        \begin{enumerate}
            \item Insert the pair $(i, uv)$ into $Q_v$.
        \end{enumerate}
    \end{enumerate}
    \item Sort all pairs $(i, uv) \in Q_v$ so the pairs are in non-increasing  order of $i$.
\end{enumerate}
For $v\in V$:
\begin{enumerate}
    \item[3.] \textsc{MatchR}$(v)$.
\end{enumerate}
Return $M$.
\end{algbox}

\begin{funcbox}{\sc MatchR($v$)}
 While $Q_v$ is not empty:
\begin{enumerate}
    \item $(j, uv) \leftarrow$ the first element of $Q_v$
    \item $j_v \leftarrow j$ \hfill {\color{gray} \# This is only used in the analysis}
    \item $\util(uv) \leftarrow w(uv) - y_u$
    \item If $\util(uv) < (1+\eps)^{j}$:
    \begin{itemize}
        \item Remove $(j,uv)$ from $Q_v$.
    \end{itemize}
    \item Else: 
    \begin{enumerate}
        \item $y_v \leftarrow \util(uv)$ \hfill {\color{gray}\# This is only used in the analysis}
        \item $y_u \leftarrow y_u + \eps \cdot (\util(uv))$ \hfill {\color{red}\# Update rule}
        \item If $u$ was matched to $v'$ so that $uv'\in M$:
        \begin{itemize}
            \item Remove $(u, v')$ from $M$
                \item $y_{v'} \leftarrow 0$ \hfill {\color{gray}\# This is only used in the analysis}
            \item Add $(u, v)$ to $M$
            \item \textsc{MatchR}$(v')$
            \item Return
        \end{itemize}
        \item Else:
        \begin{itemize}
            \item Add $(u,v)$ to $M$
            \item Return
        \end{itemize}
    \end{enumerate}
\end{enumerate}
\end{funcbox}

\myparagraph{Data structure.}
We store for each vertex $v \in V$ the list $Q_v$ as well as its currently matched edge if it exists. In the pseudocode we keep for each vertex $v$ a value $j_v$ corresponding to the highest weight threshold $(1+\eps)^{j_v}$ that we will consider.
We also keep a value $y_v$ which corresponds to the utility $v$ receives \emph{before} we update the price $y_u$ when $v$ is matched to $u$. Note that $j_v$ and $y_v$ are only needed in the analysis.

\myparagraph{Running time.}
The creation and sorting of the lists $Q_v$ takes time $O(|N(v)|(k_{max} + k_{min}))$ if we
use bucket sort as there are only $k_{max} + k_{min}$ distinct weights. The running time of all
calls to $\textsc{MatchR}(v)$ is dominated by the size of $Q_v$, as each iteration in $\textsc{MatchR}(v)$ removes an element of $Q_v$ and takes $O(1)$ time. Thus, the total time is $O\left(\sum_{v\in V} |N(v)|(k_{max} + k_{min})\right) = O(m(k_{max} + k_{min})) = O(m\eps^{-1} \log(n/\eps)).$

\myparagraph{Invariants maintained by the algorithm.}
The algorithm maintains five different invariants.
\begin{invariant}\label{inv1}
For all $v\in V$, and all $u\in N(v)$, $\util(uv)= w(uv) - y_u \le (1+\eps)^{j_v+1}$.
\end{invariant}
\begin{proof}
This clearly is true at the beginning, since $j_v$ is initialized to $k_{max}$, and \[\util(uv) = w(uv) < (1+\eps)^{j_{uv}+1}.\]
As the algorithm proceeds, $\util(uv)$ which equals $w(uv) - y_u$ only decreases as $y_u$ only increases. Thus, we only have to make sure that the condition holds whenever $j_v$ decreases.
The value $j_v$ only decreases from some value, say $j+1$, to a new value $j$, in \textsc{MatchR}$(v)$ and when this happens $Q_v$ does not contain any pairs $(j',uv)$ with $j' >j$ anymore. Thus, there does not exist a neighbor $u$ of $v$ with $\util(uv) \ge (1+\eps)^{j+1}$. It follows that
when $j_v$ decreases to $j$ for all $u \in N(v)$ it holds that $\util(uv) < (1+\eps)^{j_v+1}$.
\end{proof}

\begin{invariant} \label{inv2}
For all $u\in U$ $y_u\ge 0$ and $y_u$ never decreases over the course of the algorithm. Furthermore, if $u\in U$ is not matched, then $y_u = 0$.
\end{invariant}

\begin{proof}
We initialize $y_u$ to $0$. If $u$ is never matched, we never change $y_u$, so it stays $0$. Throughout the algorithm, we only ever increase $y_u$.
\end{proof}

\begin{invariant} \label{inv3}
For all $v\in V$ for which \textsc{MatchR}$(v)$ was called at least once,
if $v$ is unmatched, then $y_v=0$ and $Q_v$ is empty. 
Furthermore, for all $u\in N(v)$ we have that $y_u + y_v = y_u > (1-\eps) \cdot w(uv)$.
\end{invariant}
\begin{proof}
\textsc{MatchR}$(v)$ terminates (i) after it matches $v$ and recurses or (ii) if $Q_v$ is empty. Initially $v$ is unmatched and $y_v$ is set to 0. If $v$ is matched, it is possible that for some $v'\in V$, $v' \ne v$, that $v$ becomes temporarily unmatched during \textsc{MatchR}$(v')$ and $y_v$ is set to 0, but \textsc{MatchR}$(v)$ will be immediately called again.
Thus, whenever $v$ is unmatched, $y_v = 0$.

Hence, if the last call to \textsc{MatchR}$(v)$ does not result in $v$ being matched, then
this means that $Q_v$ must be empty and $y_v = 0$.
Since $Q_v$ is empty, 
then for all $u\in N(v)$,
we must have $\util(uv) < (1+\eps)^{-k_{min}} \le \eps$.
Since we rescaled weights so that $w(uv) \ge 1$, we know that $\util(uv) < \eps \le \eps \cdot w(uv)$. 
Note that $\util(uv) = w(uv) -y_u^*$ where $y_u^*$ denotes the value of $y_u$ before $u$ was matched, and $y_u \ge y_u^*$. Thus,
\[ y_u \ge y_u^* = w(uv) - \util(uv) > (1-\eps) \cdot w(uv). \]
\end{proof}

\begin{invariant} \label{inv4}
If $v\in V$ is matched to $u\in U$, then for all $u' \in N(v)$, $y_{u'} + y_{v} \ge (1-\eps)\cdot w(u'v)$ for as long as $v$ stays matched.
\end{invariant}
\begin{proof}
Note that $y_v$ doesn't change as long as $v$ stays matched, and for all $u'\in N(v)$, $y_{u'}$ can only increase by Invariant~\ref{inv2}, so it suffices to prove $y_{u'} + y_{v} \ge (1-\eps)\cdot w(u'v)$ right after $v$ was matched to $u$. 

Let $y_u^*$ be the value of $y_u$ right before $v$ was matched to $u$. 
Note that $y_v = w(uv) - y_u^*$.
For $u' = u$, we know that $y_v + y_u^* = w(uv)$, and, by Invariant~\ref{inv2}, $y_{u}^* \le y_u$ so $y_v + y_u \ge w(uv)$.

For all other $u' \in U$, 
right before we updated $y_u$, we had that
$(1+\eps)^{j_v}\le \util(uv)$ and, by Invariant~\ref{inv1}, $\util(u'v) \le (1+\eps)^{j_v+1}$.
Thus, $(1+\eps) \cdot \util(uv) \ge \util(u'v)$, so that $\util(uv) = w(uv) - y^*_u = y_v \ge (1+\eps)^{-1} \cdot (w(u'v) - y_{u'})$ .
As $y_{u'}\ge 0$ by Invariant~\ref{inv2} and  $1 > \eps > 0$, it follows that:
\begin{align*}
y_{v} + y_{u'} \ge (1+\eps)^{-1}  \cdot w(u'v) + (1 - (1+\eps)^{-1}) \cdot y_{u'} 
               \ge (1-\eps) \cdot w(u'v),
\end{align*}
\end{proof}
\begin{invariant}\label{inv5}
If $v\in V$ is matched to $u\in U$, then $y_u + y_v \le (1+\eps)\cdot w(uv)$ for as long as $v$ remains matched to $u$.
\end{invariant}
\begin{proof}
Note that $y_u$ and $y_v$ don't change as long as $v$ remains matched to $u$.
Let $y_u^*$ denote the value of $y_u$ right before the update rule of line 5(b) in \textsc{MatchR}.
Then observe $y_u = y_u^* + \eps\cdot (w(uv) - y_u^*)$, and $y_v = w(uv) - y_u^*$. Thus,
\begin{align*}
    y_u + y_v   &= y_u^* + \eps\cdot (w(uv) - y_u^*) + w(uv) - y_u^* 
    \\          &= (1+\eps) \cdot w(uv) - \eps \cdot y_u^* 
    \\          &\le (1+\eps) \cdot w(uv).
\end{align*}
\end{proof}

\myparagraph{Approximation factor.}
We will show the approximation factor of the matching $M$ found by the algorithm by primal dual analysis.
We remark that it is possible to show this result purely combinatorially as well.
We will show that this $M$ and a vector $y$ satisfy the complementary slackness condition up to a $1\pm \eps$ factor, which implies the approximation guarantee.
This was used by Duan and Pettie \cite{DuanP14} (the original lemma was for general matchings, we have specialized it here to bipartite matchings).

\begin{lemma}[Lemma 2.3 of \cite{DuanP14}] \label{lem:approxslack}
Let $M$ be a matching and let $y$ be an assignment of the dual variables.
Suppose $y$ is a complementary solution to $M$ in the following approximate sense: 
\begin{itemize}
    \item [(i)] For all $uv \in E, y_u + y_v \ge (1-\eps_0) \cdot w(uv)$,  
    \item [(ii)] For all $e\in M$, $y_u+y_v \le (1+\eps_1) \cdot w(uv)$,
    \item [(iii)] The $y$-values of all unmatched vertices are zero.
\end{itemize}
Then $M$ is a $\left((1+\eps_1)^{-1}(1-\eps_0)\right)$-approximate maximum weight matching.
\end{lemma}

\begin{proof}
Let $M^*$ be the maximum weight matching. 
\begin{align*}
w(M) &= \sum_{uv\in M} w(uv) &&
\\   &\ge \sum_{uv\in M}(1+\eps_1)^{-1} \cdot(y_u+y_v) && \text{by (ii)}
\\   &= (1+\eps_1)^{-1} \left(\sum_{u\in U} y_u+\sum_{v\in V}y_v\right) && \text{by (iii)}
\\   &\ge (1+\eps_1)^{-1} \sum_{uv\in M^*} (y_u + y_v) && \text{as $M^*$ is a matching}
\\   &\ge (1+\eps_1)^{-1} \sum_{uv\in M^*} (1-\eps_0) \cdot w(uv) && \text{by (i)}
\\   &\ge (1-\eps_0)(1+\eps_1)^{-1} w(M^*) && 
\end{align*}
\end{proof}

This lemma along with our invariants is enough for us to prove the approximation factor of our algorithm.
\begin{lemma} \label{lem:algslack}
\textsc{MultiplicativeAuction}$(G = (U\cup V, E))$ outputs a $(1-\eps')$-approximate maximum weight matching of the bipartite graph $G$.
\end{lemma}
\begin{proof}
Let $\eps > 0$ be a parameter depending on $\eps'$ that we will choose later.
We begin by choosing an assignment of the dual variables  $y_u$ for $u \in U$ and $y_v$ for $v \in V$ as exactly the values used by the algorithm at termination.
It remains to verify that we satisfy the conditions of Lemma~\ref{lem:approxslack}.
Property (i) is satisfied by Invariant~\ref{inv3} or Invariant~\ref{inv4} (depending on whether $v$ is matched or not) for $\eps_0 = \eps$.
Property (ii) is satisfied by Invariant~\ref{inv5} for $\eps_1 = \eps$.
Property (iii) is satisfied by Invariant~\ref{inv2} and Invariant~\ref{inv3}.
Thus we have satisfied Lemma~\ref{lem:approxslack} with $\eps_0 = \eps$ and $\eps_1 = \eps$. Setting $\eps = \eps'/2$ gives us a $(1-\eps')$-approximate maximum weight matching.
\end{proof}

We have shown the following result that is weaker than what we have set out to prove by a factor of $\log(n\eps^{-1})$ that we will show how to get rid of in Section~\ref{subsec:faster}.
\begin{theorem}
Let $G=(U\cup V, E)$ be a weighted biparitite graph, and $\eps$ be a value such that $1>\eps > 0$. There exists a multiplicative auction algorithm running in time $O(m\eps^{-1}\log(n\eps^{-1}))$ that finds a $(1-\eps)$-approximate maximum weight matching of $G$.
\end{theorem}

\subsection{Improving the running time}
\label{subsec:faster}

\myparagraph{Variations to the update rule}
We remark that there is some flexibility in choosing the update rule in line 5(a) of \textsc{MatchR}.

\begin{observation} \label{obs:updaterules}
To compute an $(1-\eps')$-approximate maximum weight matching
the update rule in line 5(a) of \textsc{MatchR} can be any of the following:
\begin{itemize}
    \item[(1)] $y_u\leftarrow y_u + \delta \cdot (\util(uv))$, with $\delta \le \eps'/2$,
    \item[(2)] $y_u\leftarrow y_u+ \delta \cdot (y_u)$, with $\delta \le \eps'/2$,
    \item[(3)] $y_u\leftarrow y_u + \delta \cdot (w(uv))$, with $\delta \le \eps'/4$.
\end{itemize}
\end{observation}
\begin{proof}
It suffices to verify that all invariants hold for different update rules.
Invariant~\ref{inv1}, \ref{inv2}, \ref{inv3}, and \ref{inv4} all hold regardless of the update rule, as they only use the fact that $y_u$ is non-decreasing throughout the algorithm, so we will only focus on Invariant~\ref{inv5}.

We proved that update rule (1) works in Section~\ref{subsec:slower} for $\delta = \eps = \eps'/2$.
Note that if we chose an $\delta < \eps$, we would still have $y_u \le y_u^* + \eps (w(uv)-y_u^*)$, and Invariant~\ref{inv5} holds.

To prove update rule (2) works for $\delta \le \eps = \eps'/2$,
let $v\in V$ be matched to $u\in U$, and $y_u^*$ be the value of $y_u$ right before the update rule. Observe that $y_u = (1+\delta)\cdot y_u^* \le (1+\eps)\cdot y_u^*$ and $y_v = w(uv)-y_u^*$. Furthermore $y_u^* \le w(uv)$ as otherwise $\util(uv) < 0$ and $v$ would not be trying to match to $u$. Thus, 
\[
    y_u + y_v  \le (1+\eps)\cdot y_u^* + w(uv) - y_u^*
    = w(uv) + \eps \cdot y_u^* 
    \le (1+\eps)\cdot w(uv).
\]

Since we have shown that either update rule (1) or (2) work, we can choose the larger of the two update rules, i.e. the update of adding $\delta \cdot \max\{util(uv), y_u\}$ is also a valid update rule.
However, as $\util(uv) + y_u = w(uv)$, this means that $\delta \cdot (w(uv)) \le \eps'/4 \cdot (w(uv))$, so (3) is also a valid update rule.

\end{proof}
\myparagraph{Remarks.}
Update rule (2) offers an alternative way to implement the algorithm with a running time of $O(m\eps^{-1}\log(n\eps^{-1})$.
Update rule (3) shows that $v$ can update the value of $y_u$ at most $O(\eps^{-1})$ times before $\util(uv)$ becomes non-positive, so using update rule (3) results in at most $O(m\eps^{-1})$ total updates. 
Furthermore, a careful reader may have noticed that Invariant~\ref{inv3} only requires for an edge $uv$ that $\util(uv)\le \eps \cdot w(uv)$ when we stop considering that edge, so it suffices to only consider edges in multiples of $\eps/2$ and stop considering an edge when it falls below a value of $\eps \cdot w(uv)$.

\myparagraph{Improved algorithm}
For simplicity of the exposition we will assume $2\eps^{-1}$ is a positive integer (otherwise we can choose a slightly smaller $\eps$).
To improve the running time to $O(m\eps^{-1})$, we use the observation in the above remark that every edge only needs to be updated $O(\eps^{-1})$ times if we use update rule (3) and we only need to consider edges in multiples of $\eps/2$.
Thus it suffices if we change line~1.(b) in \textsc{MultiplicativeAuction} 
to insert copies of an edge $uv$ if it has weight of the form $j \cdot w(uv)$ for some $j=\eps, 3\eps/2, 2\eps, \dots, 1$, after rounding down to the nearest power of $(1+\eps)$.
This change implies that we insert $O(\eps^{-1} |N(v)|)$ items into $Q_v$ for every $v\in V$. 
However, sorting $Q_v$ for every vertex individually, would be too slow. 

We will instead sort on all the rounded edge weights at once, as we have $O(m\eps^{-1} )$ total copies of the edges that can take on values of integers between $-k_{min}$ and $k_{max}$. 
As $k_{max}+k_{min} = O(n\eps^{-1}\log \eps^{-1}) = \poly(n)$, we can actually use radix sort to sort all the edges in linear time.
Afterwards, we can go through the weight classes in decreasing order to insert the pairs into the corresponding $Q_v$. 
We explicitly give the pseudocode below as \textsc{MultiplicativeAuction+}.

\begin{algbox}{\textsc{MultiplicativeAuction+}$(G = (U\cup V, E))$} \label{alg:MultiplicativeAuctionPlus}
$M \leftarrow \emptyset$. \\
$y_u \leftarrow 0$ for all $u \in U$.\\
$Q_v \leftarrow \emptyset$ for all $v \in V$.  \\
$L \leftarrow \emptyset$  \\
For $uv\in E$:
\begin{enumerate}
    \item For $j$ from $\eps$ to $1$ in multiples of $\eps/2$:
    \begin{enumerate}
        \item $i \leftarrow \ilog(j \cdot w(uv))$
        \item Insert the pair $(i, uv)$ into $L$.
    \end{enumerate}
\end{enumerate}
Sort $L$ in decreasing order with radix sort.\\
For $(i, uv)$ in $L$:
\begin{enumerate}
    \item[2.] Insert $(i, uv)$ to the back of $Q_v$.
\end{enumerate}
For $v\in V$:
\begin{enumerate}
    \item[3.] \textsc{MatchR}$(v)$.
\end{enumerate}
Return $M$.
\end{algbox}

\myparagraph{New runtime.}
Radix sorting all $O(m\eps^{-1})$ pairs and initializing the sorted $Q_v$ for all $v\in V$ takes linear time in the number of pairs. 
The total amount of work done in \textsc{MatchR}$(v)$ for a vertex $v\in V$ is $O(|N(v)| \eps^{-1})$ which also sums to $O(m \eps^{-1})$.
Thus we get our desired running time and have proven our main theorem that we restate here.

\staticmatching*

\section{Dynamic algorithm}

There are many monotonic properties of our static algorithm.
For instance, for all $u\in U$ the $y_u$ values strictly increase. 
As another example, for all $v\in V$ the value of $j_v$ strictly decreases.
These monotonic properties allow us to extend \textsc{MultiplicativeAuction+} to a dynamic setting with the following operations.

\decmatching*

\myparagraph{Type (1) operations: Deleting a vertex in $U$.}
To delete a vertex $u\in U$, we can mark $u$ as deleted and skip all edges $uv$ in $Q_v$ for any $v\in V$ in all further computation.
If $u$ were matched to some vertex $v\in V$, that is if there exists an edge $uv \in M$, we need to unmatch $v$ and remove $uv$ from $M$.
All our invariants hold except Invariant~\ref{inv3} for the unmatched $v$. 
To restore this invariant we simply call \textsc{MatchR}$(v)$.

\myparagraph{Type (2) operations: Adding a new vertex to $V$ with all incident edges.}
To add a new vertex $v$ to $V$ with $\ell$ incident edges to $u_1v, \dots, u_\ell v$ with $w(u_1v) > \cdots > w(u_\ell v)$, we can create the queue $Q_v$ by inserting the $O(\eps^{-1})$ pairs such that it is non-increasing in the first element of the pair.
Afterwards we call \textsc{MatchR}$(v)$.
All invariants hold after doing so.

\section*{Declarations}

\paragraph*{Funding}
\erclogowrapped{5\baselineskip}
\emph{Monika Henzinger}: This work was done in part at the University of Vienna. This project has received funding from the European Research Council (ERC) under the European Union's Horizon 2020 research and innovation programme (Grant agreement No. 101019564 ``The Design of Modern Fully Dynamic Data Structures (MoDynStruct)'' and from the Austrian Science Fund (FWF) project ``Fast Algorithms for a Reactive Network Layer (ReactNet)'', P~33775-N, with additional funding from the \textit{netidee SCIENCE Stiftung}, 2020--2024.

\paragraph*{Acknowledgements}
The first author thanks Chandra Chekuri for useful discussions about this paper.

\bibliography{ref}
\bibliographystyle{alphaurl}

\end{document}